\documentclass{llncs} 

\usepackage[noadjust]{cite}
\usepackage{amssymb,amsmath,graphics,color,upgreek,enumerate,boxedminipage}
\usepackage{array}
\usepackage{hyperref}
\usepackage{arydshln}
\usepackage{ragged2e}
\usepackage{graphicx}
\usepackage{microtype}

\newcommand{\NP}{{\sf NP}}

\newcommand{\problemdef}[3]{
	\begin{center}
		\begin{boxedminipage}{.99\textwidth}
			\textsc{{#1}}\\[2pt]
			\begin{tabular}{ r p{0.8\textwidth}}
				\textit{~~~~Instance:} & {#2}\\
				\textit{Question:} & {#3}
			\end{tabular}
		\end{boxedminipage}
	\end{center}
}

\usepackage{tikz,pgfpages}
\usetikzlibrary{arrows,shapes}

\begin{document}

\title{Filling the Complexity Gaps for Colouring\\ Planar and Bounded Degree Graphs\thanks{An extended abstract of this paper 
appeared in the proceedings of the 26th International Workshop on Combinatorial Algorithms (IWOCA 2015)~\cite{DDJP15}.
The research in this paper was supported by EPSRC (EP/K025090/1) and the Leverhulme Trust (RPG-2016-258).}}
\author{Konrad K. Dabrowski\inst{1}, Fran\c{c}ois~Dross\inst{2}, Matthew~Johnson\inst{1} \and Dani\"el~Paulusma\inst{1}}
\institute{Department of Computer Science, Durham University,\\ Lower Mountjoy, South Road, Durham DH1 3LE, United Kingdom
\texttt{\{konrad.dabrowski,matthew.johnson2,daniel.paulusma\}@durham.ac.uk}
\and 
Universit\'e C\^ote d'Azur, I3S, CNRS, Inria, France
\texttt{francois.dross@inria.fr}
}
\maketitle

\begin{abstract}
A colouring of a graph $G=(V,E)$ is a function $c: V\rightarrow\{1,2,\ldots \}$ such that $c(u)\neq c(v)$ for every $uv\in E$. A $k$-regular list assignment of~$G$ is a function~$L$ with domain~$V$ such that for every $u\in V$, $L(u)$ is a subset of $\{1, 2, \dots\}$ of size~$k$. A colouring~$c$ of~$G$ respects a $k$-regular list assignment~$L$ of~$G$ if $c(u)\in L(u)$ for every $u\in V$.
A graph~$G$ is $k$-choosable if for every $k$-regular list assignment~$L$ of~$G$, there exists a colouring of~$G$ that respects~$L$. We may also ask if for a given $k$-regular list assignment~$L$ of a given graph~$G$, there exists a colouring of~$G$ that respects~$L$. This yields the {$k$-{\sc Regular List Colouring}} problem. For $k\in \{3,4\}$ we determine a family of classes~${\cal G}$ of planar graphs, such that either $k$-{\sc Regular List Colouring} is \NP-complete for instances $(G,L)$ with $G\in {\cal G}$, or every $G\in {\cal G}$ is $k$-choosable. By using known examples of non-$3$-choosable and non-$4$-choosable graphs, this enables us to classify the complexity of $k$-{\sc Regular List Colouring} restricted to planar graphs, planar bipartite graphs, planar triangle-free graphs and to planar graphs with no $4$-cycles and no $5$-cycles. We also classify the complexity of $k$-{\sc Regular List Colouring} and a number of related colouring problems for graphs with bounded maximum degree.
\end{abstract}

\section{Introduction}

A \emph{colouring} of a graph is a labelling of its vertices such that adjacent vertices do not have the same label.
We call these labels \emph{colours}.
Graph colouring problems are central to the study of combinatorial algorithms and they have many theoretical and practical applications.
A typical problem asks whether a colouring exists under certain constraints, or how difficult it is to find such a colouring.
For example, in the {\sc List Colouring} problem, a graph is given where each vertex has a list of colours and one wants to know if the vertices can be coloured using only colours in their lists.
The {\sc Choosability} problem asks whether such list colourings are guaranteed to exist whenever all the lists have a certain size~$k$ 
(if so, then the graph is said to be $k$-choosable).
In fact, an enormous variety of colouring problems can be defined and there is now a vast collection of literature on this subject.
For longer introductions to the type of problems we consider we refer to two recent surveys~\cite{C14,GJPS}.

In this paper, we are concerned with the computational complexity of colouring problems (we give formal definitions of these problems in Section~\ref{s-terminology}).
For many such problems, the complexity is well understood in the case where we allow every graph as input, so it is natural to consider problems with restricted inputs in order to increase our understanding of their hardness.
Our two main objectives are related to this aim.
They are:
\begin{enumerate}[1.]
\item to exploit structural results in order to obtain complexity results;
\item to fill a number of complexity gaps in order to obtain complete complexity classifications.
\end{enumerate}
The graph classes we consider are the classes of planar graphs and of graphs with bounded maximum degree.
As such our paper consists of two main parts. 

In the first (Section~\ref{s-planar}), we study planar graphs and consider a natural variant of the {\sc List Colouring} problem, closely related to the {\sc Choosability} problem.
As we will discuss later in more detail, there exist many results for the latter problem restricted to planar graphs (see for example~\cite{AT92,CMR12,DLS09,Gu96,KT94,LXL99,Mo06,MRW06,Th94,Th95,Vo93,Vo95,Vo07,WLC10,WLC11,WWW08}).
Most of these results are structural.
They either show that every graph of some subclass~${\cal G}$ of planar graphs is $k$-choosable for some small value of~$k$, or construct a concrete example of a non-$k$-choosable graph that belongs to~${\cal G}$. 
When not every graph in~${\cal G}$ is $k$-choosable, it is natural to ask if the vertices of a given graph from~${\cal G}$ can be coloured in polynomial time using only colours from their lists for some given list assignment~$L$, which assigns a list of size~$k$ to each vertex.
For $k\in \{3,4\}$, we prove two general theorems, which give us a family of subclasses~${\cal G}$ of planar graphs, such that either this problem is \NP-complete for instances $(G,L)$ with $G\in {\cal G}$ or every graph in~${\cal G}$ is $k$-choosable.
We will then combine known structural choosability results with these theorems to obtain \NP-hardness results for this problem on planar graphs and on a number of subclasses of planar graphs.
This enables us to fill some complexity gaps in order to obtain a complete classification of the computational complexity of the problem for these graph classes. 

Some of the known results we use in the first part of our paper are for planar graphs with bounded degree.
In the second part of the paper (Section~\ref{s-bounded}), we combine these and other old and new results to fill some more complexity gaps and obtain complete complexity classifications for a number of colouring problems on graphs with bounded maximum degree.

We first introduce necessary definitions and terminology in Section~\ref{s-terminology}. 

\section{Preliminaries}\label{s-terminology}

A {\em colouring} of a graph $G=(V,E)$ is a function $c: V\rightarrow\{1,2,\ldots \}$ such
that $c(u)\neq c(v)$ whenever $uv\in E$.
We say that~$c(u)$ is the {\it colour} of~$u$.
For a positive integer~$k$,
if $1\leq c(u)\leq k$ for all $u\in V$, then~$c$ is a {\em $k$-colouring} of~$G$.
We say that~$G$ is {\it $k$-colourable} if a $k$-colouring of~$G$ exists.
The \textsc{Colouring} problem is to decide whether a graph~$G$ is $k$-colourable for some given integer~$k$.
If~$k$ is fixed, that is, not part of the input, we obtain the $k$-{\sc Colouring} problem.

A {\it list assignment} of a graph $G=(V,E)$ is a function~$L$ with domain~$V$ such that
for each vertex $u\in V$, $L(u)$ is a subset of $\{1, 2, \dots\}$.
This set is called the {\it list} of {\it admissible} colours for~$u$.
If $L(u)\subseteq \{1,\ldots,k\}$ for each $u\in V$, then~$L$ is a \emph{$k$-list assignment}.
The {\it size} of a list assignment~$L$ is the maximum list size~$|L(u)|$ over all vertices $u\in V$.
A colouring~$c$
{\it respects}~${L}$ if $c(u)\in L(u)$ for all $u\in V$.
Given a graph~$G$ with a $k$-list assignment~$L$, the {\sc List Colouring} problem is to decide whether~$G$
has a colouring that respects~$L$.
If~$k$ is fixed, then we have the {\sc List $k$-Colouring} problem.
Fixing the size of~$L$ to be at most~$\ell$ gives the $\ell$-{\sc List Colouring} problem.
We say that a list assignment~$L$ of a graph $G=(V,E)$ is {\it $\ell$-regular} if, for all $u\in V$, $L(u)$ contains exactly~$\ell$ colours.
This gives us the following problem, which is one focus of this paper.
It is defined for each integer $\ell\geq 1$ (note that~$\ell$ is {\it fixed}, that is, $\ell$ is not part of the input).

\problemdef{\sc $\ell$-Regular List Colouring}{a graph~$G$ with an $\ell$-regular list assignment~$L$.}
{does~$G$ have a colouring that respects~$L$?}

\noindent
A {\it $k$-precolouring} of a graph $G=(V,E)$ is a function $c_W :W\rightarrow\{1,2,\ldots,k\}$ for some subset $W\subseteq V$.
A $k$-colouring~$c$ of~$G$ is an {\it extension} of a $k$-precolouring~$c_W$ of~$G$ if $c(v)=c_W(v)$ for each $v \in W$.
Given a graph~$G$ with a precolouring~$c_W$, the {\sc Precolouring Extension} problem is to decide whether~$G$ has a $k$-colouring that extends~$c_W$.
If~$k$ is fixed, we obtain the $k$-{\sc Precolouring Extension} problem.
The relationships amongst
the problems introduced are shown in \figurename~\ref{f-col}.

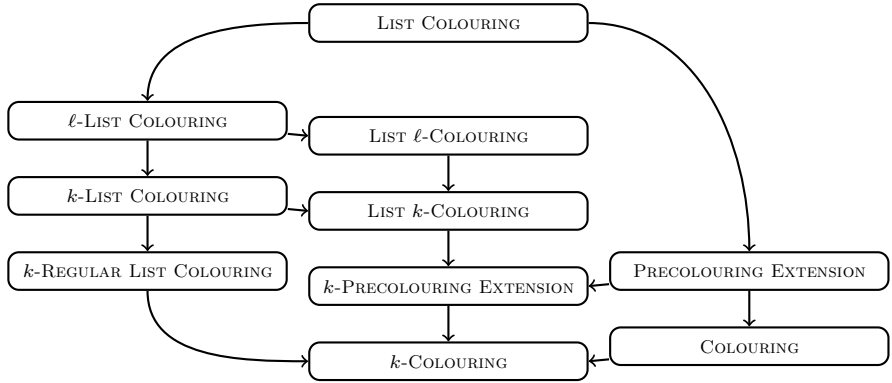
\begin{figure}[ht]
\begin{center}

\begin{tikzpicture}[problem/.style={draw,rectangle,thick,rounded corners,minimum width=3.7cm,minimum height=0.5cm,font=\scriptsize}]

\node[problem] (kregular) at (0,3) {{\sc $k$-Regular List Colouring}};
\node[problem] (llistcolouring) at (0,5) {{\sc $\ell$-List Colouring}};
\node[problem] (klistcolouring) at (0,4) {{\sc $k$-List Colouring}};
\node[problem] (listcolouring) at (4,6.3) {{\sc List Colouring}};
\node[problem] (listlcolouring) at (4,4.8) {{\sc List $\ell$-Colouring}};
\node[problem] (listkcolouring) at (4,3.8) {{\sc List $k$-Colouring}};
\node[problem] (kprecolouring) at (4,2.8) {{\sc $k$-Precolouring Extension}};
\node[problem] (kcolouring) at (4,1.8) {{\sc $k$-Colouring}};
\node[problem] (precolouring) at (8,3) {{\sc Precolouring Extension}};
\node[problem] (colouring) at (8,2) {{\sc Colouring}};

\draw[->,thick] (listcolouring) to[out=180, in=90] (llistcolouring.north);
\draw[->,thick] (listcolouring) to[out=0, in=90] (precolouring.north);
\draw[->,thick] (llistcolouring) -- (listlcolouring.west);
\draw[->,thick] (klistcolouring) -- (kregular);
\draw[->,thick] (llistcolouring) -- (klistcolouring);
\draw[->,thick] (listlcolouring) -- (listkcolouring);
\draw[->,thick] (klistcolouring) -- (listkcolouring.west);
\draw[->,thick] (precolouring) -- (colouring);
\draw[->,thick] (listkcolouring) -- (kprecolouring);
\draw[->,thick] (kprecolouring) -- (kcolouring);
\draw[->,thick] (precolouring) -- (kprecolouring.east);
\draw[->,thick] (kregular.south) to[out=270, in=180] (kcolouring.west);
\draw[->,thick] (colouring) -- (kcolouring.east);
\end{tikzpicture}
\end{center}
\caption{\label{f-col} Relationships between {\sc Colouring} and its variants.
An arrow from one problem to another indicates that the latter is (equivalent to) a special case of the former; $k$ and~$\ell$ are any two arbitrary integers for which $\ell\geq k$.
For instance, $k$-{\sc Colouring} is a special case of $k$-{\sc Regular List Colouring}. This can be seen by giving the list $L(u)=\{1,\ldots,k\}$ to each vertex~$u$ in an instance graph of {\sc Colouring}.
We recall that in the {\sc $\ell$-List Colouring} problem, the size of the lists is {\em at most}~$\ell$.
We also observe that {\sc $\ell$-Regular List Colouring} and {\sc $k$-Regular List Colouring} are not comparable for any $k \neq \ell$.}
\end{figure}

For an integer $\ell\geq 1$, a graph $G=(V,E)$ is $\ell$-{\em choosable} if, for every $\ell$-regular list assignment~$L$ of~$G$,
there exists a colouring that respects~$L$.
The corresponding decision problem is the {\sc Choosability} problem.
If~$\ell$ is fixed, we obtain the $\ell$-{\sc Choosability} problem.

We emphasize that {\sc $\ell$-Regular List Colouring} and {\sc $\ell$-Choosability} are two fundamentally different problems.
For the former we must decide whether there exists a colouring that respects a {\it particular} $\ell$-regular list assignment.
For the latter we must decide whether or not {\em every} $\ell$-regular list assignment has a colouring that respects it.
As we will see later, this difference also becomes clear from a complexity point of view: for some graph classes {\sc $\ell$-Regular List Colouring} is computationally easier than {\sc $\ell$-Choosability}, whereas, perhaps more surprisingly, for other graph classes, the reverse holds.

For two vertex-disjoint graphs~$G$ and~$H$, we let $G+H$ denote the disjoint union $(V(G)\cup V(H), E(G)\cup E(H))$, and~$kG$ denote the disjoint union of~$k$ copies of~$G$.
If~$G$ is a graph containing an edge~$e$ or a vertex~$v$ then $G-e$ and $G-v$ denote the graphs obtained from~$G$ by deleting~$e$ or~$v$, respectively.
If~$G'$ is a subgraph of~$G$ then $G-G'$ denotes the graph with vertex set~$V(G)$ and edge set $E(G) \setminus E(G')$.

We let~$C_n$, $K_n$ and~$P_n$ denote the cycle, complete graph and path on~$n$ vertices, respectively.
A wheel is a cycle with an extra vertex added that is adjacent to all other vertices.
The wheel on~$n$ vertices is denoted~$W_n$; note that $W_4=K_4$.
A graph on at least three vertices is {\it $2$-connected} if it is connected and there is no vertex whose removal disconnects it.

For a set of graphs~${\cal H}$, a graph $G$ is {\it ${\cal H}$-free} if $G$ contains no induced subgraph isomorphic to a graph in~${\cal H}$, whereas $G$ is {\it ${\cal H}$-subgraph-free} if it contains no subgraph isomorphic to a graph in~${\cal H}$.
The {\it girth} of a graph is the length of its shortest cycle.

\section{Planar Graphs}\label{s-planar}

\subsection{Known Results for Planar Graphs}\label{s-known}
We start with a classical result observed by Erd{\H{o}}s~et~al.~\cite{ERT79} and Vizing~\cite{Vi79}.

\begin{theorem}[\cite{ERT79,Vi79}]\label{t-old}
$2$-{\sc List Colouring} is polynomial-time solvable.
\end{theorem}

Garey~et~al. proved the following result, which is in contrast to
the fact that every planar graph is $4$-colourable by the Four Colour Theorem~\cite{AH89}.

\begin{theorem}[\cite{GJS74}]\label{t-gjs}
{\sc $3$-Colouring} is \NP-complete for planar graphs
of maximum degree~$4$.
\end{theorem}

Next we present results found by several authors on the existence of $k$-choosable graphs for various graph classes.

\begin{theorem} \label{t-prev-results}
The following statements hold for $k$-choosability:
\mbox{}
\begin{enumerate}[(i)]
\renewcommand{\theenumi}{\thetheorem.(\roman{enumi})}
\renewcommand{\labelenumi}{(\roman{enumi})}
\item Every planar graph is $5$-choosable \emph{\cite{Th94}}\label{t-th}.
\item There exists a planar graph that is not $4$-choosable \emph{\cite{Vo93}}\label{t-voigt}.
\item Every planar triangle-free graph is $4$-choosable \emph{\cite{KT94}}\label{t-kt}.
\item Every planar graph with no $4$-cycles is $4$-choosable \emph{\cite{LXL99}}\label{t-lss}.
\item There exists a planar triangle-free graph that is not $3$-choosable \emph{\cite{Vo95}}\label{t-voigt2}.
\item There exists a planar graph with no $4$-cycles, no $5$-cycles and no intersecting triangles that is not $3$-choosable \emph{\cite{MRW06}}\label{t-mrw}.
\item Every planar bipartite graph is $3$-choosable \emph{\cite{AT92}}\label{t-at}.
\end{enumerate}

\end{theorem}

We note that smaller examples of graphs than were used in the original proofs have been found for Theorems~\ref{t-voigt}~\cite{Gu96}, \ref{t-voigt2}~\cite{Mo06} and \ref{t-mrw}~\cite{WWW08} and that Theorem~\ref{t-mrw} strengthens a result of Voigt~\cite{Vo07}.
We recall that Thomassen~\cite{Th95} first showed that every planar graph of girth at least~$5$ is $3$-choosable, and that a number of results have since been obtained on $3$-choosability of planar graphs in which certain cycles are forbidden; see, for example,~\cite{CMR12,DLS09,WLC10,WLC11}.

We will also use the following result of Chleb\'ik and Chleb\'ikov\'a.

\begin{theorem}[\cite{CC06}]\label{t-cc}
{\sc List Colouring} is \NP-complete for $3$-regular planar bipartite graphs that have
a list assignment in which each list is one of $\{1,2\}$, $\{1,3\}$, $\{2,3\}$, $\{1,2,3\}$ and
all the neighbours of each vertex with three colours in its list have two colours in their lists.
\end{theorem}

\subsection{New Results for Planar Graphs}\label{s-our}

Theorems~\ref{t-old}--\ref{t-prev-results} have a number of immediate consequences for the complexity of {\sc $\ell$-Regular List Colouring} when
restricted to planar graphs. For instance, Theorem~\ref{t-gjs} implies that {\sc $3$-Regular List Colouring} is \NP-complete for planar graphs,
whereas Theorem~\ref{t-prev-results}.(i) shows that {\sc $5$-Regular List Colouring} is polynomial-time solvable on this graph class.
In fact, we notice a complexity jump from being \NP-complete to being trivial (that is, the answer is always yes) when~$\ell$ changes from~$3$ to~$5$. 
It is a natural question to determine the complexity for the missing case $\ell=4$.

In this section we settle this missing case and also present  a number of new hardness results for {\sc $\ell$-Regular List Colouring} restricted to various subclasses of planar graphs.  At the end of this section we show how to combine the known results with our new ones to obtain a number of dichotomy results (Corollaries~\ref{c-main1}--\ref{c-main0}). We deduce some of our new results from two more general theorems, namely Theorems~\ref{t-general1} and~\ref{t-general}, which we state below, but which we prove in Sections~\ref{s-proof1} and~\ref{s-proof2}, respectively.

\begin{sloppypar}
\begin{theorem}\label{t-general1}
Let~${\cal H}$ be a finite set of $2$-connected planar graphs.
Then $4$-{\sc Regular List Colouring} is \NP-complete for planar ${\cal H}$-subgraph-free graphs if
there exists a planar ${\cal H}$-subgraph-free graph that is not $4$-choosable.
\end{theorem}
\end{sloppypar}

Note that the class of ${\cal H}$-subgraph-free graphs is contained in the class of ${\cal H}$-free graphs. Hence, whenever a problem is \NP-complete for ${\cal H}$-subgraph-free graphs, it is also \NP-complete for ${\cal H}$-free graphs.

An alternative formulation of Theorem~\ref{t-general1} is that, for every finite set~${\cal H}$ of $2$-connected planar graphs,
either every pair $(G,L)$, where~$G$ is a planar ${\cal H}$-subgraph-free graph, is a yes-instance of  $4$-{\sc Regular List Colouring}, or 
$4$-{\sc Regular List Colouring} is \NP-complete for planar ${\cal H}$-subgraph-free graphs.
Results of the same flavour were proved for {\sc $3$-Colouring}, {\sc Acyclic $3$-Colouring}, {\sc $(1,0)$-Colouring} and $C_{2p+1}$-{\sc Colouring} restricted to classes of planar graphs by Esperet~et~al.~\cite{EMOP13}.
To give a more recent example of such a result, Dross, Montassier and Pinlou~\cite{DMP16} proved that either every triangle-free planar graph is near-bipartite -- that is, can be decomposed into an independent set and a forest -- or recognizing near-bipartite 
triangle-free 
planar graphs is \NP-complete. 
By using their construction and the existence of a planar graph that is not near-bipartite (for instance~$K_4$),
Bonamy~et~al.~\cite[see arXiv version]{BDFJP17} observed that the problem of recognizing near-bipartite graphs is \NP-complete on planar graphs.

We will exploit the power of Theorem~\ref{t-general1} by combining it with Theorem~\ref{t-prev-results}.
For instance, combining Theorem~\ref{t-general1} with Theorem~\ref{t-voigt} yields that $4$-{\sc Regular List Colouring} is \NP-complete for planar graphs.
We are not aware of any paper proving this result although it seems to have been known:
the result is mentioned by Thomassen~\cite{Th97} without proof, and Dvo\v{r}\'ak and Thomas~\cite{DT14} mistakenly attribute it to~\cite{Gu96}. 
We can sharpen this result as follows.
By the Four Colour Theorem~\cite{AH89}, every planar graph is $4$-colourable, that is, has a colouring respecting the list assignment that assigns list of colours $\{1,2,3,4\}$ to every vertex of the graph.
In contrast, Voigt and Wirth~\cite{VW97} showed the existence of a planar graph that does not allow a colouring respecting some specific list assignment~$L$, in which each list contains four distinct colours from $\{1,2,3,4,5\}$.
We use their example to prove the following result in Section~\ref{s-proof1}.

\begin{theorem}\label{t-main1}
$4$-{\sc Regular List Colouring} is \NP-complete for planar graphs even if every list contains four colours from $\{1,2,3,4,5\}$.
\end{theorem}

Theorem~\ref{t-general1} has further applications and can also be combined with other results from the literature.  
For instance, consider the non-$4$-choosable planar graph~$H$ from the proof of Theorem~1.7 in~\cite{Gu96}.
It can be observed that~$H$ is $W_p$-subgraph-free for all $p\geq 8$.  Wheels are $2$-connected and planar. Hence, if~${\cal H}$ is any finite set of wheels on at least eight vertices, then $4$-{\sc Regular List Colouring} is \NP-complete for planar ${\cal H}$-subgraph-free graphs.

Our basic idea for proving Theorem~\ref{t-general1} is similar to the proof technique used in~\cite{DMP16,EMOP13}.
We pick a minimal counterexample~$H$ with list assignment~$L$. 
We select an ``appropriate'' edge $e=uv$ and consider the graph $F'=F-e$.
We reduce from an appropriate colouring problem restricted to planar graphs and use copies of~$F'$ as a gadget to ensure that we can enforce a regular list assignment.
The proof of the next theorem also uses this idea.

\begin{sloppypar}
\begin{theorem}\label{t-general}
Let~${\cal H}$ be a finite set of $2$-connected planar graphs.
Then $3$-{\sc Regular List Colouring} is \NP-complete for planar ${\cal H}$-subgraph-free graphs if there exists a planar ${\cal H}$-subgraph-free graph that is not $3$-choosable.
\end{theorem}
\end{sloppypar}
Just like Theorem~\ref{t-general1},
Theorem~\ref{t-general} has a number of applications.
For instance, if we let ${\cal H}=\{K_3\}$ then Theorem~\ref{t-general}, combined with Theorem~\ref{t-voigt2}, leads to the following result (Dvo\v{r}\'ak and Kawarabayashi also briefly mentioned how to obtain this result in their paper~\cite{DK13} but do not provide a full proof).

\begin{sloppypar}
\begin{corollary}\label{t-main2}
$3$-{\sc Regular List Colouring} is \NP-complete for planar triangle-free graphs.
\end{corollary}
\end{sloppypar}
Theorem~\ref{t-general} can also be used for other classes of graphs.
For example, let~${\cal H}$ be a finite set of graphs, each of which includes a $2$-connected graph on at least five vertices as a subgraph.
Let~${\cal I}$ be the set of these $2$-connected graphs.
The graph~$K_4$ is a planar ${\cal I}$-subgraph-free graph that is not $3$-choosable (since it is not $3$-colourable).
Therefore, Theorem~\ref{t-general} implies that $3$-{\sc Regular List Colouring} is \NP-complete for planar ${\cal H}$-subgraph-free graphs.
We can obtain more hardness results by taking some other planar graph that is not $3$-choosable, such as a wheel on an even number of vertices.
Also, if we let ${\cal H}=\{C_4, C_5\}$ we can use Theorem~\ref{t-general} by combining it with Theorem~\ref{t-mrw} to find that
$3$-{\sc Regular List Colouring} is \NP-complete for planar graphs with no $4$-cycles and no $5$-cycles.
We strengthen this result as follows (see Section~\ref{s-proof4} for the proof).

\begin{theorem}\label{t-main4}
$3$-{\sc Regular List Colouring} is \NP-complete for planar graphs with no $4$-cycles, no $5$-cycles and no intersecting triangles.
\end{theorem}

Theorem~\ref{t-main1}, Corollary~\ref{t-main2} and Theorem~\ref{t-main4} can be seen as strengthenings of Theorems~\ref{t-voigt},~\ref{t-voigt2} and~\ref{t-mrw}, respectively.
Moreover, they complement Theorem~\ref{t-gjs}, which implies that $3$-{\sc List Colouring} is \NP-complete for planar graphs, and a result of Kratochv\'{\i}l~\cite{Kr93} that, for planar bipartite graphs, $3$-{\sc Precolouring Extension} is \NP-complete.
Theorem~\ref{t-main1} and Corollary~\ref{t-main2} also complement results of
Gutner~\cite{Gu96} who showed that {\sc $3$-Choosability} and {\sc $4$-Choosability} are $\Uppi_2^p$-complete for planar triangle-free graphs and planar graphs, respectively.
However, we emphasize that, for special graph classes, it is not necessarily the case that {\sc $\ell$-Choosability} is computationally harder than {\sc $\ell$-Regular List Colouring}.
For instance, contrast the fact that {\sc Choosability} is polynomial-time solvable on
$3P_1$-free graphs~\cite{GHHP13} with our next result, which we prove in Section~\ref{s-proof10}.

\begin{theorem}\label{t-3p1}
{\sc $3$-Regular List Colouring} is \NP-complete for $(3P_1,P_1+P_2)$-free graphs.
\end{theorem}

\subsection{Closing Complexity Gaps for Planar Graphs}

Our new results, combined with known results, close a number of complexity gaps for the $\ell$-{\sc Regular List Colouring} problem.
Combining Theorem~\ref{t-main1} with Theorems~\ref{t-old},~\ref{t-gjs} and~\ref{t-th} gives us Corollary~\ref{c-main1}.
Combining Theorem~\ref{t-main4} with Theorems~\ref{t-old} and~\ref{t-lss} gives us Corollary~\ref{c-main3}.
Combining Corollary~\ref{t-main2} with Theorems~\ref{t-old} and~\ref{t-kt} gives us Corollary~\ref{c-main2}, whereas
Theorems~\ref{t-old} and~\ref{t-at} imply Corollary~\ref{c-main0}.

\begin{corollary}\label{c-main1}
Let~$\ell$ be a positive integer.
Then $\ell$-{\sc Regular List Colouring}, restricted to planar graphs, is \NP-complete if $\ell\in \{3,4\}$ and
polynomial-time solvable otherwise.
\end{corollary}

\begin{corollary}\label{c-main3}
Let~$\ell$ be a positive integer.
Then $\ell$-{\sc Regular List Colouring}, restricted to planar graphs with no $4$-cycles and no $5$-cycles and no intersecting triangles, is \NP-complete if $\ell=3$ and
polynomial-time solvable otherwise (even if we allow intersecting triangles and $5$-cycles).
\end{corollary}

\begin{corollary}\label{c-main2}
Let~$\ell$ be a positive integer.
Then $\ell$-{\sc Regular List Colouring}, restricted to planar triangle-free graphs, is \NP-complete if $\ell=3$ and
polynomial-time solvable otherwise.
\end{corollary}

\begin{corollary}\label{c-main0}
Let~$\ell$ be a positive integer.
Then $\ell$-{\sc Regular List Colouring}, restricted to planar bipartite graphs, is polynomial-time solvable.
\end{corollary}

\section{Bounded Degree Graphs}\label{s-bounded}

\subsection{Known Results for Bounded Degree Graphs}\label{s-bknown}

First we present a result of Kratochv\'il and Tuza~\cite{KT94}.

\begin{theorem}[\cite{KT94}]\label{t-kt94}
{\sc List Colouring} is polynomial-time solvable on graphs of maximum degree at most~$2$.
\end{theorem}

Brooks' Theorem~\cite{Br41} states that every graph~$G$ with maximum degree~$d$ has a $d$-colouring unless~$G$ is a complete graph or a cycle with an odd number of vertices. The next result of Vizing~\cite{Vi76} generalizes Brooks' Theorem to list colourings.

\begin{theorem}[\cite{Vi76}]\label{t-vizing}
Let~$d$ be a positive integer.
Let $G=(V,E)$ be a connected graph of maximum degree at most~$d$ and let~$L$ be a $d$-regular list assignment for~$G$.
If~$G$ is not a cycle or a complete graph then~$G$ has a colouring that respects~$L$.
\end{theorem}

And we need two other results of Chleb\'ik and Chleb\'ikov\'a~\cite{CC06}.

\begin{theorem}[\cite{CC06}]\label{t-cc2}
{\sc Precolouring Extension} is polynomial-time solvable on graphs of maximum degree~$3$.
\end{theorem}

\begin{theorem}[\cite{CC06}]\label{t-precol}
Let~$k$ be a positive integer.
Then {\sc $k$-Precolouring Extension} is polynomial-time solvable for graphs of maximum degree at most~$k$.
\end{theorem}

\subsection{Closing Complexity Gaps for Bounded Degree Graphs}

We have the following two classifications, which we obtain by combining previously known results with some small observations.

\begin{corollary}\label{c-colouringdegree}
Let~$d$ be a positive integer.
The following two statements hold for graphs of maximum degree at most~$d$.
\begin{enumerate}[(i)]
\item {\sc List Colouring} is \NP-complete if $d\geq 3$ and polynomial-time solvable if $d\leq 2$.
\item {\sc Precolouring Extension} and {\sc Colouring} are \NP-complete if $d\geq 4$ and polynomial-time solvable if $d\leq 3$.
\end{enumerate}
\end{corollary}

\begin{proof}
We first consider~(i).
If $d\geq 3$, we use Theorem~\ref{t-cc}.
If $d\leq 2$, we use Theorem~\ref{t-kt94}.
We now consider~(ii).
If $d\geq 4$, we use Theorem~\ref{t-gjs}.
If $d\leq 3$, we use Theorem~\ref{t-cc2}.\qed
\end{proof}

\begin{corollary}\label{c-colouringdegree2}
Let~$d$ and~$k$ be two positive integers.
The following two statements hold for graphs of maximum degree at most~$d$.
\begin{enumerate}[(i)]
\item {\sc $k$-List Colouring} and {\sc List $k$-Colouring} are \NP-complete if $k\geq 3$ and $d\geq 3$ and polynomial-time solvable otherwise.
\begin{sloppypar}
\item {\sc $k$-Regular List Colouring} and
$k$-{\sc Precolouring Extension} 
are \NP-complete if $k\geq 3$ and $d\geq k+1$ and polynomial-time solvable otherwise.
\end{sloppypar}
\end{enumerate}
\end{corollary}

\begin{proof}
We first consider~(i).
If $k\geq 3$ and $d\geq 3$, we use Theorem~\ref{t-cc}.
If $k\leq 2$ or $d\leq 2$, we use Theorems~\ref{t-old} or~\ref{t-kt94}, respectively.

We now consider~(ii).
We start with the hardness cases and so let $k\geq 3$ and $d\geq k+1$.

First consider $k$-{\sc Precolouring Extension}. 
Theorem~\ref{t-gjs} implies that $3$-{\sc Colouring} is \NP-complete for graphs of maximum degree at most~$d$ for all $d\geq 4$.
The $k=3$ case follows immediately from this result. 
Suppose $k\geq 4$ and $d \geq k+1$. 
Consider a graph~$G$ of maximum degree~$4$.
For each vertex~$v$, we add $k-3$ new vertices $x^v_1,\ldots,x^v_{k-3}$ and edges $vx^v_1,\ldots,vx^v_{k-3}$.
Let~$G'$ be the resulting graph.
Note that~$G'$ has maximum degree at most $4+k-3=k+1\leq d$.
We define a precolouring~$c$ on the newly added vertices by assigning colour $i+3$ to each~$x^v_i$.
Then~$G'$ has a $k$-colouring extending~$c$ if and only if~$G$ has a $3$-colouring.

Now consider {\sc $k$-Regular List Colouring}.
The $k=3$ case follows immediately from Theorem~\ref{t-gjs}.
Suppose $k\geq 4$ and $d \geq k+1$. 
Consider a graph~$G$ of maximum degree~$4$.
We define the list $L(v)=\{1,\ldots,k\}$ for each vertex $v\in V(G)$.
For each vertex~$v$, we add $k-3$ new vertices $x^v_1,\ldots,x^v_{k-3}$ and edges $vx^v_1,\ldots,vx^v_{k-3}$.
We define the list $L(x^v_i)=\{i,k+1,k+2,\ldots,2k-1\}$ for each~$x^v_i$.
For each vertex~$x^v_i$, we also add~$k$ new vertices $w_1(x^v_i),\ldots,w_k(x^v_i)$ and edges such that $x^v_i,w_1(x^v_i),\ldots,w_k(x^v_i)$ form a clique (on $k+1$ vertices).
We define the list $L(w_j(x^v_i))=\{k+1,\ldots,2k\}$ for each~$w_j(x^v_i)$.
Let~$G'$ be the resulting graph. Note that~$G'$ has maximum degree at most $k+1$ and that the resulting list assignment~$L$ is a $k$-regular list assignment of~$G'$. 
Then~$G'$ has a $k$-colouring respecting~$L$ if and only if~$G$ has a $3$-colouring.

We continue with the polynomial-time solvable cases.
If $k\leq 2$, the result follows from Theorem~\ref{t-old}.
Now suppose that $k\geq 3$ and $d\leq k$.
The result for {\sc $k$-Precolouring Extension} follows from Theorem~\ref{t-precol}.
We now consider the {\sc $k$-Regular List Colouring} problem.
Let $(G,L)$ be an instance.
Since {\sc $k$-Regular List Colouring} can be solved component-wise, we may assume that~$G$ is connected.
If~$G$ is a cycle, we use Theorem~\ref{t-kt94}.
If~$G$ is a complete graph, then we use the folklore result that {\sc List Colouring} is polynomial-time solvable on complete graphs (see~\cite{GPS14} for a proof).
In all other cases we apply Theorem~\ref{t-vizing}.\qed
\end{proof}

Note that Corollary~\ref{c-colouringdegree2} does not contain a dichotomy for $k$-{\sc Colouring} restricted to graphs of maximum degree at most~$d$.
A full classification of this problem is open, but a number of results are known.
Molloy and Reed~\cite{MR14} classified the complexity for all pairs $(k,d)$ for sufficiently large~$d$.
Emden-Weinert~et~al.~\cite{EHK98} proved that {\sc $k$-Colouring} is \NP-complete for graphs of maximum degree at most~$k+\lceil\sqrt{k}\rceil-1$. 
It follows from Brooks' Theorem~\cite{Br41} that for every integer $k\geq 1$, {\sc $k$-Colouring} is polynomial-time solvable for graphs of maximum degree~$k$.
Combining this observation with the result of~\cite{EHK98} means that the smallest open case is when $k=5$ and $d=6$.

\section{Proofs}\label{s-proofs}

In Section~\ref{s-proof1} we prove Theorems~\ref{t-general1} and~\ref{t-main1}.
In Section~\ref{s-proof2} we prove Theorem~\ref{t-general}, whereas in Section~\ref{s-proof4} we prove Theorem~\ref{t-main4}, and in Section~\ref{s-proof10} we prove Theorem~\ref{t-3p1}.

\subsection{The Proofs of Theorems~\ref{t-general1} and~\ref{t-main1}}\label{s-proof1}

We need an additional result.

\begin{theorem}\label{t-new}
For every integer~$p\geq 3$, {\sc $3$-List Colouring} is \NP-complete for planar graphs of girth at least~$p$ that have
a list assignment in which each list is one of $\{1,2\}$, $\{1,3\}$, $\{2,3\}$, $\{1,2,3\}$.
\end{theorem}

\begin{proof}
By Theorem~\ref{t-cc}, {\sc List Colouring} is \NP-complete for $3$-regular planar bipartite graphs that have
a list assignment in which each list is one of $\{1,2\}$, $\{1,3\}$, $\{2,3\}$, $\{1,2,3\}$ and
all the neighbours of each vertex with three colours in its list have two colours in their lists.
We modify the hardness construction as follows.
Note that for each edge at least one of the incident vertices has a list of size~$2$.
We replace each edge by a path on an odd number of edges in such a way that the girth of the graph obtained is at least~$p$.
The new vertices on the path are all given the same list of size~$2$, identical to the list on one or other of the end-vertices.  It is readily seen that these modifications do not affect whether or not the graph can be coloured.\qed
\end{proof}

We are now ready to prove Theorem~\ref{t-general1}, which we restate below.

\medskip
\noindent
{\bf Theorem~\ref{t-general1} (restated).}
{\it Let~${\cal H}$ be a finite set of $2$-connected planar graphs.
Then $4$-{\sc Regular List Colouring} is \NP-complete for planar ${\cal H}$-subgraph-free graphs if
there exists a planar ${\cal H}$-subgraph-free graph that is not $4$-choosable.}

\begin{proof}
The problem is readily seen to be in \NP. Let~$F$ be a planar ${\cal H}$-subgraph-free graph with a $4$-regular list assignment~$L$ such that~$F$ has no colouring respecting~$L$. We may assume that~$F$ is minimal (with respect to the subgraph relation). In particular, this means that~$F$ is connected. Let~$r$ be the length of a longest cycle in any graph of~${\cal H}$. We reduce from the problem of $3$-{\sc List Colouring} restricted to planar graphs of girth at least $r+1$ in which each vertex has list
$\{1,2\}$, $\{1,3\}$, $\{2,3\}$ or $\{1,2,3\}$.
This problem is \NP-complete by Theorem~\ref{t-new}.
Let a graph~$G$ and list assignment~$L_G$ be an instance of this problem.
We will construct a planar ${\cal H}$-subgraph-free graph~$G'$ with a $4$-regular list assignment~$L'$ such that~$G$ has a colouring that respects~$L_G$ if and only if~$G'$ has a colouring that respects~$L'$.

If every pair of adjacent vertices in~$F$ has the same list, then the problem of finding a colouring that respects~$L$ is just the problem of finding a $4$-colouring which, by the Four Colour Theorem~\cite{AH89}, we know is possible.
Thus we may assume that, on the contrary, there is an edge $e=uv$ such that $|L(u)\cap L(v)|\leq 3$.
Let $F'=F-e$.
Then, by minimality, $F'$ has at least one colouring respecting~$L$, and moreover,
for any colouring of~$F'$ that respects~$L$, $u$ and~$v$ are coloured alike (otherwise we would have a colouring of~$F$ that respects~$L$).
Let~$T$ be the set of possible colours that can be used on~$u$ and~$v$ in colourings of~$F'$ that respect~$L$ and let $t=|T|$.
As $T \subseteq L(u)\cap L(v)$, we have $1 \leq t \leq 3$.
Up to renaming the colours in~$L$, we can build copies of~$F'$ with $4$-regular list assignments such that 
\begin{enumerate}[(i)]
\item the set $T$ is any given list of colours of size~$t$, and 
\item the vertex corresponding to~$u$ has any given list of~$4$ colours containing~$T$.
\end{enumerate}
We will implicitly make use of this several times in the remainder of the proof.

We say that a vertex~$w$ in~$G$ is a {\em bivertex} or {\em trivertex} if~$|L_G(w)|$ is~$2$ or~$3$, respectively.  We construct a planar ${\cal H}$-subgraph-free graph~$G'$ and a $4$-regular list assignment~$L'$ as follows.

First suppose that $t=1$.
For each bivertex~$w$ in~$G$, we do as follows.
We add two copies of~$F'$ to~$G$, which we label~$F_1(w)$ and~$F_2(w)$. 
The vertex in~$F_i(w)$ corresponding to~$u$ is labelled~$u_i^w$ for $i \in \{1,2\}$ and we set $U^w=\{u_1^w,u_2^w\}$.
We add the edges~$wu_1^w$ and~$wu_2^w$. We give list assignments to the vertices of~$F_1(w)$ and~$F_2(w)$ such that $T = \{4\}$ for~$F_1$ and $T = \{5\}$ for~$F_2$. We let $L'(w) = L_G(w) \cup \{4,5\}$.
For each trivertex~$w$ in~$G$, we do as follows.
We add one copy of~$F'$ to~$G$, which we label~$F_1(w)$. 
The vertex in~$F_1(w)$ corresponding to~$u$ is labelled~$u_1^w$ and we set $U^w=\{u_1^w\}$.
We add the edge~$wu_1^w$. We give list assignments to vertices of $F_1(w)$ such that $T = \{4\}$ for $F_1$. We let $L'(w) = L_G(w) \cup \{4\}$.
This completes the construction of~$G'$ and~$L'$ when $t=1$.

Now suppose that $t=2$.
Let $s = r$ if~$r$ is even and $s = r+1$ if~$r$ is odd (so~$s$ is even in both cases).
For each bivertex~$w$ in~$G$, we do as follows.
We add a copy of~$F'$ to~$G$, which we label~$F_1(w)$, and identify the vertex in~$F_1(w)$ corresponding to~$u$ with~$w$. 
We give list assignments to vertices of $F_1(w)$ such that $T = L_G(w)$ 
and $L'(w) = L_G(w) \cup \{4,5\}$. 
For each trivertex~$w$ in~$G$, we do as follows.
We add~$s$ copies of~$F'$ to~$G$ which we label~$F_i(w)$, $1 \leq i \leq s$.
The vertex in~$F_i(w)$ corresponding to~$u$ is labelled~$u_i^w$.
Let $U^w=\{u^w_i \mid 1 \leq i \leq s\}$.
Add edges such that
the union of~$w$ and~$U^w$ induces a cycle on $s+1$ vertices.
For all $1 \leq i \leq s$, we give list assignments to vertices of~$F_i(w)$ such that $T = \{4,5\}$. 
We let $L'(w) = \{1,2,3,4\}$.
This completes the construction of~$G'$ and~$L'$ when $t=2$.

Now suppose that $t=3$.
For each bivertex~$w$ in~$G$, we do as follows.
We add two copies of~$F'$ to~$G$ which we label~$F_1(w)$ and~$F_2(w)$,
such that  for $i \in \{1,2\}$, the vertex in~$F_i(w)$ 
corresponding to~$u$ 
is
identified with~$w$.
We give list assignments to vertices of~$F_1(w)$ and~$F_2(w)$ such that $T = L_G(w) \cup \{4\}$ for~$F_1(w)$, $T = L_G(w) \cup \{5\}$ for~$F_2(w)$ and $L'(w) = L_G(w) \cup \{4,5\}$.
For each trivertex~$w$ in~$G$, we do as follows.
We add a copy of~$F'$ to~$G$ which we label~$F_1(w)$, such that the vertex in~$F_1(w)$ corresponding to~$u$ is 
identified with~$w$.
We give list assignments to the vertices of~$F_1(w)$ such that $T = \{1,2,3\}$ and $L'(w) = \{1,2,3,4\}$.
This completes the construction of~$G'$ and~$L'$ when $t=3$.

Note that~$G'$ is planar.
Suppose that there is a subgraph~$H$ in~$G'$ that is isomorphic to a graph of~${\cal H}$.
Since~$F$ is ${\cal H}$-subgraph-free, and since~$F'$ is obtained from~$F$ by removing one edge, $F'$ is also ${\cal H}$-subgraph-free.
Therefore for all~$w$, $H$ is not fully contained in any~$F_i(w)$.
Since~$H$ is $2$-connected and since for all~$w$ only one vertex of any~$F_i(w)$ has a neighbour outside of~$F_i(w)$, we find that~$H$ has at most one vertex in each~$F_i(w)$.
In particular, $H$ cannot contain any vertex of any~$F_i(w)$ in which the vertex corresponding to~$u$ has been attached to~$w$ (as opposed to being identified with~$w$); this includes the case when the union of~$w$ and~$U^w$ induces a cycle on $s+1$ vertices.
Hence we have found that~$H$ is a subgraph of~$G$, which contradicts the fact that~$G$ has girth at least $r+1$.
Therefore~$G'$ is ${\cal H}$-subgraph-free.

Note that in any colouring of~$G'$ that respects~$L'$, each copy of~$F'$ must be coloured such that the vertices corresponding to~$u$ and~$v$ have the same colour, which must be one of the colours from the corresponding set~$T$.
If $t=1$ and~$w$ is a trivertex, this means that the unique neighbour of~$w$ in~$U^w$ must be coloured with colour~$4$, so~$w$ cannot be coloured with colour~$4$.
Similarly, if $t=1$ and~$w$ is a bivertex or $t=2$ and~$w$ is a trivertex then the two neighbours of~$w$ in~$U^w$ must be coloured with colours~$4$ and~$5$, so~$w$ cannot be coloured with colours~$4$ or~$5$.
If $t=2$ and~$w$ is a bivertex or $t=3$ and~$w$ is a trivertex then~$w$ belongs to a copy of~$F'$ with $T=L_G(w)$, so~$w$ cannot have colour~$4$ or~$5$.
If $t=3$ and~$w$ is a bivertex then~$w$ belongs to two copies of~$F'$, one with $T=L_G(w) \cup \{4\}$ and one with $T= L_G(w) \cup \{5\}$.
Therefore, $w$ must be coloured with a colour from the intersection of these two sets, that is it must be coloured with a colour from~$L_G(w)$.
Therefore none of the vertices of~$G$ can be coloured~$4$ or~$5$.
Thus the problem of finding a colouring of~$G'$ that respects~$L'$ is equivalent to the problem of finding a colouring of~$G$ that respects~$L_G$.
This completes the proof.\qed
\end{proof}

We will now prove Theorem~\ref{t-main1}.

\medskip
\noindent
{\bf Theorem~\ref{t-main1} (restated).}
{\it $4$-{\sc Regular List Colouring} is \NP-complete for planar graphs even if every list contains four colours from $\{1,2,3,4,5\}$.}

\begin{proof}
Recall that there exists  a planar graph~$F^*$ with a $4$-regular list assignment~$L$ in which each list~$L(u)$ contains four colours from
$\{1,2,3,4,5\}$ such that~$F^*$ has no colouring respecting~$L$~\cite{VW97}.
We may assume without loss of generality that~$F^*$ is minimal (with respect to the subgraph relation) and use~$F^*$ as the graph~$F$ in the proof of Theorem~\ref{t-general1}.
This means that we give each vertex of the graph~$G'$ in the proof of Theorem~\ref{t-general1}  a list of colours from $\{1,2,3,4,5\}$.
The result follows.\qed
\end{proof}

\subsection{The Proof of Theorem~\ref{t-general}}\label{s-proof2}

\noindent
{\bf Theorem~\ref{t-general} (restated).}
{\it
Let~${\cal H}$ be a finite set of $2$-connected planar graphs.
Then $3$-{\sc Regular List Colouring} is \NP-complete for planar
${\cal H}$-subgraph-free graphs if there exists a planar ${\cal H}$-subgraph-free graph that is
not $3$-choosable.}

\begin{proof}
The problem is readily seen to be in \NP.
Every graph in~${\cal H}$ is $2$-connected, and therefore contains a cycle.
Let~$r$ be the length of a longest cycle in any graph of~${\cal H}$.
By assumption, there exists a planar graph~$F$ and $3$-regular list assignment~$L$ such that~$F$ is ${\cal H}$-subgraph-free and has no colouring respecting~$L$. We may assume that~$F$ is minimal by removing edges and vertices until any further removal would give a graph with a colouring respecting~$L$. In particular, we note that~$F$ is connected.

We distinguish two cases.

\medskip
\noindent
{\bf Case~1.}~$L(v)$ is the same for every vertex~$v$ in~$F$.\\
Then we may assume without loss of generality that $L(v) = \{1,2,3\}$ for all~$v$.
We reduce from {\sc $3$-Colouring} which is \NP-complete even for planar graphs by Theorem~\ref{t-gjs}.
Let~$G$ be a planar graph.
We will construct a planar ${\cal H}$-subgraph-free graph~$G'$ as follows.

Let $e = uv$ be an edge of~$F$.
Let $F' = F-e$.
Then, by minimality, $F'$ has at least one colouring respecting~$L$, which must be a $3$-colouring as every list consists of the colours $1,2,3$.
For every $3$-colouring~$c$ of~$F'$, it holds that $c(u) = c(v)$ (otherwise~$c$ would be a colouring of~$F$ that respects~$L$).
Moreover, since we can permute the colours, there is such a colouring~$c$ that colours~$u$ (and thus~$v$) with colour~$i$ for each $i\in\{1,2,3\}$.
Note that in~$F'$ the vertices~$u$ and~$v$ must be at distance at least~$2$ from each other.

Let $s = \left\lceil \frac{r}{6} \right \rceil$.
Assume that the vertices of~$G$ are ordered.
For each edge $xy\in E(G)$ with $x < y$, we do the following:
\begin{enumerate}[(i)]
\item delete~$xy$;
\item add~$s$ copies of~$F'$ labelled $F_1(xy), \ldots, F_s(xy)$ and, for $1 \leq i \leq s$, let~$u_i^{xy}$ and~$v_i^{xy}$ be the vertices in~$F_i(xy)$ corresponding to~$u$ and~$v$;
\item identify~$x$ with~$u_1^{xy}$ and, for $1 \leq i \leq s-1$, identify~$v_i^{xy}$ with~$u_{i+1}^{xy}$;
\item add an edge from~$v_s^{xy}$ to~$y$.
\end{enumerate}
Let~$G'$ be the obtained graph and note that~$G'$ is planar.
Every cycle in~$G'$ that is not contained in a copy of~$F'$ has length at least $3(2s + 1) \ge r+1$, since it corresponds to a cycle in~$G$ and in which every edge has been replaced by~$s$ successive copies of~$F'$ plus an edge.

Suppose that there is a subgraph~$H$ in~$G'$ that is isomorphic to a graph of~${\cal H}$.
Since~$F$ is ${\cal H}$-subgraph-free, and since~$F'$ is obtained from~$F$ by removing one edge, $F'$ must also be ${\cal H}$-subgraph-free.
Therefore~$H$ is not fully contained in a copy of~$F'$.
Since~$H$ is $2$-connected, this implies that there is a cycle in~$H$ that is not fully contained in a copy of~$F'$.
By definition, this cycle has length at most~$r$, a contradiction.

Suppose the graph~$G'$ has a $3$-colouring~$c$.
For each copy of~$F'$, the vertices corresponding to~$u$ and~$v$ must be coloured the same.
For all edges~$xy$ with $x < y$ in~$G$, there is a vertex in~$G'$ coloured the same as~$x$ in~$G$ that is adjacent to~$y$ in~$G'$, so $c(x) \ne c(y)$.  Therefore~$c$ restricted to~$V(G)$ is a $3$-colouring of~$G$.

On the other hand, suppose the graph~$G$ has a $3$-colouring~$c$.
We can extend this $3$-colouring to~$G'$ by doing the following: for all edges~$xy$ with $x < y$ in~$G$, colour every~$F_i(xy)$ in such a way that the vertex corresponding to~$u$ and the vertex corresponding to~$v$ have colour~$c(x)$.
This leads to a $3$-colouring of~$G'$.

\medskip
\noindent
{\bf Case~2.}~$F$ contains two vertices~$u$ and~$v$ with $L(u)\neq L(v)$.\\
As~$F$ is connected, we assume without loss of generality that~$u$ and~$v$ are adjacent; let $e=uv$.

We reduce from the problem of $3$-{\sc List Colouring} restricted to planar graphs of girth at least $r+1$ in which each vertex has list
$\{1,2\}$, $\{1,3\}$, $\{2,3\}$ or $\{1,2,3\}$.
This problem is \NP-complete by Theorem~\ref{t-new}.
Let a graph~$G$ and list assignment~$L_G$ be an instance of this problem.
We will construct a planar ${\cal H}$-subgraph-free graph~$G'$ with a $3$-regular list assignment~$L'$ such that~$G$ has a colouring that respects~$L_G$ if and only if~$G'$ has a colouring that respects~$L'$.

We define $F'=F-e$.
Then, by minimality, $F'$ has at least one colouring respecting~$L$, and moreover,
for any colouring of~$F'$ that respects~$L$, $u$ and~$v$ are coloured alike (otherwise we would have a colouring of~$F$ that respects~$L$).
Let~$T$ be the set of possible colours that can be used on~$u$ and~$v$ in colourings of~$F'$ that respect~$L$ and let $t=|T|$.
As $T \subseteq L(u)\cap L(v)$, we have $1 \leq t \leq 2$.
Let us assume, without loss of generality, that $T \subseteq \{4,5\}$ and that $4 \in T$.

We say that a vertex~$w$ in~$G$ is a bivertex or trivertex if~$|L(w)|$ is~$2$ or~$3$, respectively.
We construct a planar ${\cal H}$-free graph~$G'$.

First suppose that $t=1$.
For each bivertex~$w$ in~$G$, we do as follows.
We add a copy of~$F'$ to~$G$ which we label~$F(w)$.
The vertex in~$F(w)$ corresponding to~$u$ is labelled~$u^w$ and we set $U^w=\{u^w\}$.
We add the edge~$wu^w$.
This completes the construction of~$G'$ when $t=1$.

Now suppose that $t=2$.
Let $s = r$ if~$r$ is even and $s = r+1$ if~$r$ is odd (so~$s$ is even in both cases).
For each bivertex~$w$ in~$G$, we add~$s$ copies of~$F'$ to~$G$ which we label~$F_i(w)$, $1 \leq i \leq s$.
The vertex in~$F_i(w)$ corresponding to~$u$ is labelled~$u_i^w$.
Let $U^w=\{u^w_i \mid 1 \leq i \leq s\}$.
Add edges such that, for each bivertex~$w$ in~$G$, the union of~$w$ and~$U^w$ induces a cycle on $s+1$ vertices.
This completes the construction of~$G'$ when $t=2$.

Note that~$G'$ is planar, since it is made of planar graphs ($G$ and copies of~$F'$) connected in a way that does not obstruct planarity.
Suppose that there is a subgraph~$H$ in~$G'$ that is isomorphic to a graph of~${\cal H}$.
Since~$F$ is ${\cal H}$-subgraph-free, and since~$F'$ is obtained from~$F$ by removing one edge, $F'$ is also ${\cal H}$-subgraph-free.
Therefore for all~$w$, $H$ is not fully contained in~$F(w)$.
Since~$H$ is $2$-connected and since for all~$w$, only one vertex of~$F(w)$ has a neighbour outside of~$F(w)$, we find that~$H$ has at most one vertex in each~$F(w)$.
This means that~$H$ is a subgraph of~$G$, which contradicts the fact that~$G$ has girth at least $r+1$.
Therefore~$G'$ is ${\cal H}$-subgraph-free.

Now we define a list assignment~$L'$.
We give the vertices of each copy of~$F'$ the same lists as their corresponding vertices in~$F$, and
for each bivertex~$w$ in~$G$, we define $L'(w)=L_G(w)\cup \{4\}$, and for each trivertex~$w$ in~$G$, we define $L'(w)=L_G(w)$.
This gives us the $3$-regular list assignment~$L'$ of~$G'$.

The graph $G' - G$ has a colouring that respects (the restriction of)~$L'$ and we notice that in such a colouring each copy of~$F'$ must be coloured in such a way that, for each bivertex~$w$ in~$G$, the~$t$ vertices of~$U^w$ that are adjacent to~$w$ are coloured with the~$t$ colours of~$T$.
So one of the neighbours of~$w$ in~$U^w$ must be coloured~$4$.
Thus the problem of finding a colouring of~$G'$ that respects~$L'$ is equivalent to the problem of finding a colouring of~$G$ that respects~$L_G$.
The proof is complete.\qed
\end{proof}

\subsection{The Proof of Theorem~\ref{t-main4}}\label{s-proof4}

Theorem~\ref{t-main4} is not quite implied by Theorem~\ref{t-general}.
However, we can adapt the proof of Theorem~\ref{t-general} to prove Theorem~\ref{t-main4}, which we restate below.

\medskip
\noindent
{\bf Theorem~\ref{t-main4} (restated).}
{\it
$3$-{\sc Regular List Colouring} is \NP-complete for planar graphs with no $4$-cycles, no $5$-cycles and no intersecting triangles.}

\begin{proof}
Let~$B$ be the graph on five vertices with two triangles sharing exactly one vertex (this graph is known as the {\it butterfly}).
Consider the previous proof with ${\cal{H}} = \{C_4,C_5,B\}$.
Note that the only problem is that~$B$ is not $2$-connected.

In~\cite{MRW06} an example is given 
of a graph with no $4$-cycle, no $5$-cycle and no intersecting triangles that is not $3$-choosable.
We omit the details of this construction, as the existence of such a graph is sufficient for our proof.
Hence, let~$I$ be the example of~\cite{MRW06},
with~$L$ a $3$-regular list assignment such that there is no colouring of~$I$ respecting~$L$.
We remove edges and vertices from~$I$ until any further removal would give a graph with a colouring respecting~$L$.
This leads to a connected graph~$F$.
There are no four vertices in~$I$ with the same list inducing a connected subgraph, so there are no such four vertices in~$F$.
Therefore in~$F$ there is an edge~$uv$ such that $L(u) \ne L(v)$ (otherwise~$F$ would have at most three vertices, and thus would be $3$-choosable).

Therefore we can skip Case~$1$ and directly adapt the proof in Case~$2$.
Note that the only thing to prove is that in this case~$G'$ does not contain a subgraph isomorphic to~$B$.
Suppose~$H$ is such a subgraph.
Since~$F'$ is $B$-subgraph-free, $H$ cannot be fully contained in~$F(w)$ for any~$w$.
Since no vertex is in two different~$F(w)$
and no vertex of~$F(w)$ has two adjacent neighbours outside~$F(w)$,
 this implies that there is a triangle in~$G$, which is impossible since~$G$ has girth at least~$6$.\qed
\end{proof}

\subsection{The Proof of Theorem~\ref{t-3p1}}\label{s-proof10}

The proof is obtained by a modification of the \NP-hardness construction of $3$-{\sc List Colouring} for $(3P_1,P_1+P_2)$-free graphs from~\cite{GPS14}.
Recall that we included this result in our paper to illustrate that {\sc $k$-Choosability} and {\sc $k$-Regular List Colouring} can have different complexities when restricted to special graph classes. 
Indeed, since {\sc Choosability} is polynomial-time solvable on $3P_1$-free graphs~\cite{GHHP13}, Theorem~\ref{t-3p1} shows that {\sc $k$-Choosability} may even be easier than {\sc $k$-Regular List Colouring}.

\medskip
\noindent
{\bf Theorem~\ref{t-3p1} (restated).}
{\it  {\sc $3$-Regular List Colouring} is \NP-complete for $(3P_1,P_1+P_2)$-free graphs.}

\begin{proof}
The problem is readily seen to belong to~\NP. 
Golovach~et~al.~\cite{GHHP13} showed that $3$-{\sc List Colouring} is \NP-complete for $(3P_1,P_1+P_2)$-free graphs in which every vertex has a list of size~$2$ or~$3$.
Let~$G$ be such an instance. We add three new vertices $s, t, u$ to~$G$. We make 
$s,t,u$ adjacent to each other and to each original vertex of~$G$. This results in a $(3P_1,P_1+P_2)$-free graph~$G'$.
We take three new colours $1,2,3$ and set $L(s)=L(t)=L(u)=\{1,2,3\}$. This forces colour~$1$ to be used to colour one of $s,t$ or~$u$.
Then all that remains is to add colour~$1$ to the list of every vertex of~$G$ that has a list of size~$2$.\qed
\end{proof}

\section{Conclusions}
As well as presenting some dichotomies for a number of colouring problems for graphs with bounded maximum degree, we have given several dichotomies for the $k$-{\sc Regular List Colouring} problem restricted to subclasses of planar graphs. 
In particular we showed \NP-hardness of the cases $k=3$ and $k=4$ restricted to planar~${\cal H}$-subgraph-free graphs for several sets~${\cal H}$ of $2$-connected planar graphs.
Our method implies that for such sets~${\cal H}$ it suffices to find a counterexample to $3$-choosability or to $4$-choosability, respectively.
It is a natural to ask whether we can determine the complexity of $3$-{\sc Regular List Colouring} and $4$-{\sc Regular List Colouring} for any class of planar ${\cal H}$-subgraph-free graphs. 
However, we point out that even when restricting~${\cal H}$ to be a finite set of $2$-connected planar graphs, this would be very hard (and beyond the scope of this paper) as it would require solving several long-standing conjectures in the literature.
For example, when~${\cal H} = \{C_4,C_5,C_6\}$, Montassier~\cite{Mo06} conjectured that every planar ${\cal H}$-subgraph-free graph is $3$-choosable.

A drawback of our method is that we need the set of graphs~${\cal H}$ to be $2$-connected. 
If we forbid a set~${\cal H}$ of graphs that are not $2$-connected, the distinction between polynomial-time solvable and \NP-complete cases is not clear, and both cases may occur even if we forbid only one graph. We illustrate this below with an example.

\medskip
\noindent
{\it Example.} Let~${\cal H}$ contain only the star~$K_{1,r}$ for some $r\geq 2$.
Note that~$K_{1,r}$-subgraph-free graphs are exactly those graphs that have maximum degree at most $r-1$. 
Hence, if $r=3$, then  $3$-{\sc Regular List Colouring} is polynomial-time solvable due to Theorem~\ref{t-kt94}.
However, there exist larger values of~$r$ for which the problem is \NP-complete. In order to see this we adapt the
proof of 
Theorem~\ref{t-general}. The hardness reductions in this proof multiply the maximum degree of our instances by some constant~$d$ that is at most the maximum degree of the no-instance~$F$.
By Theorems~\ref{t-gjs} and~\ref{t-cc}, the problems we reduce from are \NP-complete even for graphs with maximum degree at most~$4$.
Hence, we have proven the following:
if~${\cal H}$ is a finite set of $2$-connected planar graphs and~$F$ is a non-$3$-choosable planar ${\cal H}$-subgraph-free graph with maximum degree~$d$, then $3$-{\sc Regular List Colouring} is \NP-complete on planar ${\cal H}$-subgraph-free graphs with maximum degree at most~$4d$.
We can take $F=K_4$ to deduce that $3$-{\sc Regular List Colouring} is \NP-complete on planar $K_{1,13}$-subgraph-free graphs.

\medskip
\noindent
{\it Acknowledgements.} 
We thank Zden\v{e}k Dvo\v{r}\'ak and Steven Kelk for helpful comments on earlier versions of this paper.

\end{document}